\documentclass[DIV13,11pt]{scrartcl}

\usepackage[noblocks]{authblk}

\usepackage{enumerate}

\setkomafont{title}{\normalfont \LARGE \bfseries}
\setkomafont{section}{\normalfont \Large \bfseries \boldmath}
\setkomafont{subsection}{\normalfont \large \bfseries}
\setkomafont{subsubsection}{\normalfont \normalsize \bfseries}
\setkomafont{descriptionlabel}{\itshape}
\setkomafont{paragraph}{\normalfont \bfseries \boldmath}

\usepackage[OT1]{fontenc}
\usepackage[english]{babel}
\usepackage[utf8]{inputenc}
\usepackage{charter}
\usepackage[svgnames]{xcolor}
\usepackage{amssymb,amsfonts,mathrsfs,yhmath}
\usepackage[amsthm,amsmath,thmmarks]{ntheorem}
\usepackage{enumitem}
\usepackage{expdlist,graphicx}
\usepackage{booktabs}
\usepackage{caption}
\usepackage[labelformat=simple]{subcaption}

\usepackage{epsfig}
\usepackage[numbers]{natbib}
\usepackage[linkcolor=Blue,colorlinks=true,citecolor=ForestGreen,filecolor=red,urlcolor=Blue]{hyperref}
\usepackage{cleveref}[2011/12/24]
\usepackage{enumitem,linegoal}

\usepackage{tikz}
\usetikzlibrary{arrows}
\usetikzlibrary{arrows.meta}
\usetikzlibrary{decorations.pathmorphing}
\usetikzlibrary{decorations.markings}
\usetikzlibrary{calc}

\tikzset{
  cir/.style = {circle,draw,fill,inner sep=.7pt},
  circ/.style = {circle,draw,fill,inner sep=1.3pt},
  circg/.style = {circle,draw=lightgray,fill=lightgray,inner sep=1.3pt},
  circr/.style = {circle,draw=red,fill=red,inner sep=1.3pt},
  invisible/.style = {circle,draw=none,inner sep=0pt,font=\tiny},
  nonedge/.style={decorate,decoration={snake,amplitude=.3mm,segment length=1mm},draw}
}

\theoremstyle{plain}
\newtheorem{theorem}{Theorem}
\newtheorem{corollary}[theorem]{Corollary}
\newtheorem{lemma}[theorem]{Lemma}

\newtheorem{observation}[theorem]{Observation}
\newtheorem{claim}{Claim}
\theoremstyle{definition}
\newtheorem{definition}{Definition}

\theoremstyle{remark}
\newtheorem{remark}[theorem]{Remark}

\usepackage[pagewise,displaymath,mathlines]{lineno}
\makeatletter
\newcommand{\eqnum}{\leavevmode\hfill\refstepcounter{equation}\textup{\tagform@{\theequation}}}
\makeatother

\date{}

\title{\Large Approximating Independent Set and Dominating Set on VPG graphs}
\author[1]{\normalsize Esther Galby\thanks{esther.galby@unifr.ch}}
\author[2]{\normalsize Andrea Munaro\thanks{a.munaro@qub.ac.uk}}
\affil[1]{Department of Informatics, University of Fribourg, Switzerland}
\affil[2]{School of Mathematics and Physics, Queen’s University Belfast, United Kingdom}

\begin{document}
\maketitle

\begin{abstract}
\textbf{Abstract.} We consider \textsc{Independent Set} and \textsc{Dominating Set} restricted to VPG graphs (or, equivalently, string graphs). 
We show that they both remain $\mathsf{NP}$-hard on $B_0$-VPG graphs admitting a representation such that each grid-edge belongs to at most one path and each horizontal path has length at most two. On the other hand, combining the well-known Baker's shifting technique with bounded mim-width arguments, we provide simple PTASes on VPG graphs admitting a representation such that each grid-edge belongs to at most $t$ paths and the length of the horizontal part of each path is at most $c$, for some $c \geq 1$. 
\end{abstract}

%\end{frontmatter}

\section{Introduction}\label{intro}

Given a family $O$ of geometric objects in the plane, the intersection graph of $O$ has the objects in $O$ as its vertices and two vertices $o_i, o_j \in O$ are adjacent in the graph if and only if $o_i \cap o_j \neq \varnothing$. If $O$ is a set of curves in the plane, where a \textit{curve} is a subset of $\mathbb{R}^2$ homeomorphic to the unit interval $[0, 1]$, then the intersection graph of $O$ is a \textit{string graph}. Many important graph classes like planar graphs and chordal graphs are subclasses of string graphs \citep{GIP18,Kra911} and so it is natural to study classic optimization problems such as \textsc{Independent Set} and \textsc{Dominating Set} on string graphs. 

\citet{asinowski} introduced the class of \textit{Vertex intersection graphs of Paths on a Grid} (\textit{VPG graphs} for short). A graph is a \textit{VPG graph} if one can associate a path on a grid with each vertex such that two vertices are adjacent if and only if the corresponding paths intersect on at least one grid-point. It is not difficult to see that the class of VPG graphs coincides with that of string graphs \citep{asinowski}. If every path in the VPG representation has at most $k$ \textit{bends} i.e., $90$ degrees turns at a grid-point, the graph is a \textit{$B_k$-VPG graph} and a \textit{segment} of a path is a vertical or horizontal line segment in the polygonal curve constituting the path. We remark that $B_0$-VPG graphs are also known as $2$-DIR (see, e.g., \citep{KN90}). 

\citet{GLS09} introduced the class of \textit{Edge intersection graphs of Paths
on a Grid} (\textit{EPG graphs} for short) as those graphs for which there exists a collection of paths on a grid in one-to-one correspondence with their vertex set, two vertices being adjacent if and only if the corresponding paths intersect on at least one grid-edge. It turns out that every graph is EPG \citep{GLS09} and $B_{k}$-EPG graphs have been defined similarly to $B_{k}$-VPG graphs. Notice that $B_0$-EPG graphs are the well-known interval graphs.

\textsc{Independent Set} is known to be $\mathsf{NP}$-complete on $B_k$-VPG graphs even for $k = 0$ \citep{KN90}. Therefore, there has been a focus on providing approximation algorithms for restricted subclasses of string graphs. \citet{FP11} gave, for every $\varepsilon > 0$, a $n^{\varepsilon}$-approximation algorithm for \textit{$k$-string graphs} i.e., string graphs in which every two curves intersect each other at most $k$ times. \citet{LMS15} provided a $O(\log^2 n)$-approximation algorithm for $B_1$-VPG graphs and a $O(\log d)$-approximation algorithm for equilateral $B_1$-VPG graphs (a $B_1$-VPG graph is \textit{equilateral} if, for each path, its horizontal and vertical segment have the same length), where $d$ denotes the ratio between the maximum and minimum length of segments of paths. Finally, they showed that \textsc{Independent Set} on equilateral $B_1$-VPG graphs where each horizontal and vertical segment have length $1$ is $\mathsf{NP}$-complete. Improving on \citep{LMS15}, \citet{BD17} provided a $4\log n$-approximation algorithm for the weighted version of \textsc{Independent Set} on $B_{2}$-VPG graphs. The idea is to partition a $B_{2}$-VPG graph into $O(\log n)$ outer-string graphs and then solve the problem optimally on each of them by \citep{KMPV17}. In the case of $B_{1}$-VPG, this result was further improved by \citet{BC19}, who provided a $4\max\{1, \log\mathsf{OPT}\}$-approximation algorithm for the weighted version of \textsc{Independent Set}. They also showed that \textsc{Independent Set} can be solved in $O(n^2)$ time for graphs admitting a \textit{grounded} string representation (i.e. a string representation in which one endpoint of each string is attached to a grounding line  and all strings lie on one side of the line), where the strings are $y$-monotone (not necessarily strict) polygonal paths, the length of each string is bounded by a constant and all the bends and endpoints are on integral coordinates. 
\citet{Meh181} considered the weighted version of \textsc{Independent Set} on $B_k$-VPG graphs for which the longest segment among all segments of paths in the graph has length $c$, for some $c > 0$ (not required to be a constant), and provided a $(ck + c + 1)$-approximation algorithm. Notice that, to the best of our knowledge, it is not known whether there exists a constant-factor approximation algorithm for \textsc{Independent Set} even on $B_1$-VPG graphs.

Concerning EPG graphs, \citet{EGM13} showed that \textsc{Independent Set} is $\mathsf{NP}$-complete on $B_1$-EPG graphs and provided a $4$-approximation algorithm. \citet{BBGP15} showed that the problem admits no PTAS on $B_{1}$-EPG graphs, unless $\mathsf{P} = \mathsf{NP}$, even if each path has its vertical segment or its horizontal segment of length at most $3$ and that it remains $\mathsf{NP}$-hard on $B_{1}$-EPG graphs even if all the paths have their horizontal segment and vertical segment of length at most $2$. On the other hand, they provided a PTAS for $B_{1}$-EPG graphs such that each path has its horizontal segment of length at most $c$, for some fixed constant $c$. This was done by adapting the well-known Baker's shifting technique \citep{Bak94}.

Let us now review \textsc{Dominating Set}. Observe first that it is $\mathsf{APX}$-hard on $1$-string $B_1$-VPG graphs. Indeed, every circle graph (i.e. intersection graph of chords in a circle) is a $1$-string $B_1$-VPG graph \citep{asinowski} and \textsc{Dominating Set} is $\mathsf{APX}$-hard on circle graphs \citep{DP06}. \citet{Meh18} considered the subclass of $1$-string $B_1$-VPG graphs in which no endpoints of a path belong to any other path and provided an $O(1)$-approximation algorithm. \citet{BMMS19} considered intersection graphs of L-frames, where an \textit{L-frame} is a path on a grid with exactly one bend, and provided a $(2+\varepsilon)$-approximation algorithm in the case the bend of each path belongs to a diagonal line with slope $-1$. They also showed that the problem is $\mathsf{APX}$-hard if each L-frame intersects a diagonal line and that the same holds if instead all the frames intersect a vertical line. \citet{CDM19} provided an $8$-approximation algorithm on intersection graphs of L-frames intersecting a common vertical line. They also showed that there is an $O(k^4)$-approximation algorithm on unit $B_k$-VPG graphs\footnote{A $B_{k}$-VPG graph is \textit{unit} if each path consists only of segments with unit length. Notice that every $B_k$-VPG graph is a unit $B_{k'}$-VPG graph for some finite $k' \geq k$.} and, on the negative side, that the problem is $\mathsf{NP}$-hard on unit $B_1$-VPG graphs. Similarly to \textsc{Independent Set}, it is not known whether there exists a constant-factor approximation algorithm for \textsc{Dominating Set} on $B_1$-VPG graphs.

Concerning EPG graphs, \citet{BMMS19}  showed that \textsc{Dominating Set} for $B_1$-EPG graphs is hard to approximate within a factor of $1.1377$ even if all the paths intersect a vertical line. 

\subsection{Our results} 

As we have just mentioned, three natural constraints on VPG graphs have been considered in the search for efficient algorithms for \textsc{Independent Set} and \textsc{Dominating Set}: bound the number of bends on each path, bound the number of intersections between any two paths and bound the lengths of segments of paths. Unfortunately, combining these constraints is not enough to guarantee polynomial-time solvability: In \Cref{sechard}, we show that \textsc{Independent Set} and \textsc{Dominating Set} remain $\mathsf{NP}$-complete when restricted to $1$-string $B_0$-VPG graphs such that each horizontal path has length at most $2$, even if the representation is part of the input.     

But what about approximation algorithms? Perhaps surprisingly, it turns out that the problems admit PTASes when those constraints are in place. More precisely, in \Cref{ptases}, we provide PTASes for \textsc{Independent Set} and \textsc{Dominating Set} when restricted to VPG graphs admitting a representation $\mathcal{R} = (\mathcal{G}, \mathcal{P})$ such that:
\begin{enumerate}
\item\label{bends}each path in $\mathcal{P}$ has a polynomial (in $\vert \mathcal{P} \vert$) number of bends;
\item\label{edgepath} each grid-edge in $\mathcal{G}$ belongs to at most $t$ paths in $\mathcal{P}$;
\item\label{horiz} the horizontal part of each path in $\mathcal{P}$ has length at most $c$.
\end{enumerate}
Here the \textit{horizontal part} of a path is the interval corresponding to the projection of the path onto the horizontal axis.

Clearly, for fixed $k \geq 0$, $B_{k}$-VPG graphs satisfy condition \ref{bends}. The class of VPG graphs satisfying \ref{edgepath} is rich as well: it contains $k$-string VPG graphs, for any fixed $k$, VPG graphs with maximum degree at most $t-1$ and VPG graphs with maximum clique size at most $t$.

Notice that recognizing string graphs (and so VPG graphs) is $\mathsf{NP}$-complete \citep{Kra91,SSS03}. Similarly, for each fixed $k \geq 0$, recognizing $B_{k}$-VPG graphs is $\mathsf{NP}$-complete \citep{CJKV12,Kra94}. Therefore, in our PTASes, we assume that a representation of a VPG graph is always given as part of the input. The reason behind condition \ref{bends} is to avoid the following pathological behavior: there exist string graphs on $n$ vertices requiring paths with $2^{\Omega(n)}$ bends in any representation (this follows from \citep{KM91}).

Our result on \textsc{Dominating Set} is best possible in the sense that, if we remove one of conditions \ref{edgepath} and \ref{horiz}, the problem does not admit a PTAS unless $\mathsf{P} = \mathsf{NP}$ (\Cref{best}). The situation is more subtle for \textsc{Independent Set}, as it has been asked several times whether the problem is $\mathsf{APX}$-hard on $B_{k}$-VPG graphs (see, e.g., \citep{BD17,Meh17}) and this remains open.

As observed above, the study of these two problems on VPG graphs has focused mostly on $B_{k}$-VPG graphs with $k \leq 2$ and, to the best of our knowledge, ours are the first PTASes on a non-trivial subclass of VPG graphs. The PTAS for \textsc{Dominating Set} shows that the constant-factor approximation algorithm on unit $B_{k}$-VPG graphs in \citep{CDM19} can in fact be improved if input graphs satisfy also condition \ref{edgepath}. Our PTASes are obtained by adapting Baker's shifting technique \citep{Bak94} to the string graph setting and showing boundedness of an appropriate width parameter. Baker's technique has already been applied to \textsc{Independent Set} in $B_{1}$-EPG graphs \citep{BBGP15} and our main contribution is to pair it with powerful mim-width arguments. Mim-width is a graph parameter, introduced by \citet{Vat12}, measuring how easy it is to decompose a graph along vertex cuts inducing a bipartite graph with small maximum induced matching size. Combining results in \citep{BV13,BTV13}, it is known that $(\sigma, \rho)$-domination problems (a class of graph problems including \textsc{Independent Set} and \textsc{Dominating Set} introduced by \citet{TP97}) can be solved in $O(n^{w})$ time, assuming a branch decomposition of mim-width $w$ is provided as part of the input. Our key observation is that if a VPG graph admits a VPG representation with a bounded number of columns and such that each grid-edge belongs to a bounded number of paths, then the graph has bounded mim-width (\Cref{lem:boundedmimvpg}). This result is of independent interest and best possible, in the sense that both conditions on the VPG graph are needed to guarantee boundedness of mim-width (\Cref{bestposs}). We then use Baker's shifting technique by solving the problems optimally on each slice with bounded number of columns (\Cref{algogrid}).

\section{Preliminaries}\label{prel}

In this paper we consider only finite simple graphs. Given a graph $G$, we usually denote its vertex set by $V(G)$ and its edge set by $E(G)$. If $G'$ is a subgraph of $G$ and $G'$ contains all the edges of $G$ with both endpoints in $V(G')$,
then $G'$ is an induced subgraph of $G$ and we write $G' = G[V(G')]$.

\vspace{0.1cm}

\textbf{Neighborhoods and degrees.} For a vertex $v \in V(G)$, the \textit{closed neighborhood} $N_{G}[v]$ is the set of vertices adjacent to $v$ in $G$ together with $v$. The \textit{degree} $d_{G}(v)$ of a vertex $v \in V(G)$ is the number of edges incident to $v$ in $G$. A \textit{$k$-vertex} is a vertex of degree $k$ and a \textit{$k^{+}$-vertex} is a vertex of degree at least $k$. We refer to a $3$-vertex as a \textit{cubic} vertex and to a $0$-vertex as an \textit{isolated} vertex. The \textit{maximum degree} $\Delta(G)$ of $G$ is the quantity $\max\left\{d_{G}(v): v \in V\right\}$ and $G$ is \textit{subcubic} if $\Delta(G) \leq 3$. 

\vspace{0.1cm}

\textbf{Graph operations.} Given a graph $G = (V, E)$ and $V' \subseteq V$, the operation of \textit{deleting the set of vertices $V'$} from $G$ results in the graph $G - V' = G[V\setminus V']$. A \textit{$k$-subdivision} of an edge $e \in E(G)$ is the operation replacing $e$ with a path of length $k + 1$.
     
\vspace{0.1cm}  

\textbf{Graph classes and special graphs.} A graph is \textit{$Z$-free} if it does not contain induced subgraphs isomorphic to graphs in a set $Z$. A \textit{complete graph} is a graph whose vertices are pairwise adjacent and the complete graph on $n$ vertices is denoted by $K_{n}$. A \textit{triangle} is the graph $K_{3}$. A graph is bipartite if its vertex set admits a partition into two classes such that every edge has its endpoints in different classes. A \textit{split graph} is a graph whose vertices can be partitioned into a clique and an independent set. A \textit{caterpillar} is a tree whose non-leaf vertices form a path.

\vspace{0.1cm}
           
\textbf{Graph properties and parameters.} A set of vertices or edges of a graph is \textit{minimum} with respect to the property $\mathcal{P}$ if it has minimum size among all subsets having property $\mathcal{P}$. The term \textit{maximum} is defined analogously. An \textit{independent set} of a graph is a set of pairwise non-adjacent vertices.  The size of a maximum independent set of $G$ is denoted by $\alpha(G)$. A \textit{clique} of a graph is a set of pairwise adjacent vertices. A \textit{dominating set} of $G$ is a subset $D \subseteq V(G)$ such that each vertex in $V(G) \setminus D$ is adjacent to a vertex in $D$. The size of a minimum dominating set of $G$ is denoted by $\gamma(G)$. A \textit{matching} of a graph is a set of pairwise non-incident edges. 
An \textit{induced matching} in a graph is a matching $M$ such that no two vertices belonging to different edges in $M$ are adjacent in the graph. 

\vspace{0.1cm}
           
\textbf{VPG, CPG and EPG graphs.} Given a rectangular grid $\mathcal{G}$, its horizontal lines are referred to as \textit{rows} and its vertical lines as \textit{columns}. The grid-point lying on row $x$ and column $y$ is denoted by $(x,y)$. 

A graph $G$ is \textit{VPG} if there exists a collection $\mathcal{P}$ of paths on a grid $\mathcal{G}$ such that $\mathcal{P}$ is in one-to-one correspondence with $V(G)$ and two vertices are adjacent in $G$ if and only if the corresponding paths intersect. A graph $G$ is \textit{CPG} if there exists a collection $\mathcal{P}$ of interiorly disjoint paths on a grid $\mathcal{G}$ such that $\mathcal{P}$ is in one-to-one correspondence with $V(G)$ and two vertices are adjacent in $G$ if and only if the corresponding paths touch. A graph $G$ is \textit{EPG} if there exists a collection $\mathcal{P}$ of paths on a grid $\mathcal{G}$ such that $\mathcal{P}$ is in one-to-one correspondence with $V(G)$ and two vertices are adjacent in $G$ if and only if the corresponding paths share a grid-edge. 

A VPG graph is a $B_{k}$-VPG graph if there exists a collection $\mathcal{P}$ as above such that every path in $\mathcal{P}$ has at most $k$ bends i.e., $90$ degree turns at a grid-point. For a VPG graph $G$, the pair $\mathcal{R} = (\mathcal{G},\mathcal{P})$ is a \textit{VPG representation} of $G$ and, more specifically, a \textit{$B_k$-VPG representation} if every path in $\mathcal{P}$ has at most $k$ bends. The path in $\mathcal{P}$ corresponding to the vertex $u$ is denoted by $P_u$. For a path $P \in \mathcal{P}$, we denote by $\partial(P)$ the set of endpoints of $P$. The \textit{length} of a path $P \in \mathcal{P}$ is the number of grid-edges used and an \textit{interior point} of $P$ is a point belonging to $P$ and different from its endpoints. A \textit{bend-point} of $P$ is a grid-point corresponding to a bend of $P$ and a \textit{segment} of $P$ is either a vertical or horizontal line segment in the polygonal curve constituting $P$. Analogous definitions hold for CPG and EPG graphs.

We now need to describe how a VPG or CPG-representation $\mathcal{R} = (\mathcal{G},\mathcal{P})$ of a graph $G$ is encoded. For the grid $\mathcal{G}$, we only keep track of the \textit{grid-step $\sigma$}, that is, the length of a grid-edge. For each path $P \in \mathcal{P}$, we have three sequences of points in $\mathbb{R}^2$. The first sequence $s(P) = (x_1,y_1),  (x_2,y_2), \dots, (x_{\ell_P},y_{\ell_P})$ consists of the endpoints $(x_1,y_1)$ and $(x_{\ell_P},y_{\ell_P})$ of $P$ and all the bend-points of $P$ in their order of appearance while traversing $P$ from $(x_1,y_1)$ to $(x_{\ell_P},y_{\ell_P})$. In other words, for any $i \in [\ell_P -1]$, $[(x_i,y_i),(x_{i+1},y_{i+1})]$ is a segment of $P$. The second and third sequences contain, for any $P' \neq P$, the points $f_{P,P'}$ and $l_{P,P'}$ (if any) corresponding to the first intersection of $P$ with $P'$ and to the last intersection of $P$ with $P'$, respectively, i.e there is no intersection with $P'$ on the portion of $P$ from $(x_1,y_1)$ to $f_{P,P'}$ and on the portion of $P$ from $l_{P,P'}$ to $(x_{\ell_P},y_{\ell_P})$. The intersection points in $\{f_{P,P'} : P' \neq P\}$ are ordered according to their appearance while traversing $P$ from $(x_1,y_1)$ to $(x_{\ell_P},y_{\ell_P})$, whereas the intersection points in $\{l_{P,P'} : P' \neq P\}$ are ordered according to their appearance while traversing $P$ from $(x_{\ell_P},y_{\ell_P})$ to $(x_1,y_1)$. The \textit{first intersection point $f_P$ of $P$} is the smallest intersection point in $\{f_{P,P'} : P' \neq P\}$ according to the above order and the \textit{last intersection point $l_P$ of $P$} is the smallest intersection point in $\{l_{P,P'} : P' \neq P\}$ according to the above order. Notice that if each path in $\mathcal{P}$ has a number of bends bounded by a polynomial in $|V(G)|$, then the size of this data structure is polynomial in $|V(G)|$. Moreover, knowing $s(P)$ for each $P \in \mathcal{P}$, we can compute in time polynomial in $|V(G)|$ the second and third sequences of each path. Given $s(P)$, we can also easily determine the horizontal part $h(P)$ of the path $P$ as follows. Let $x^P_{\min} = \min \{x_i : i \in [\ell_P]\}$ and let $x^P_{\max} = \max \{x_i : i \in [\ell_P]\}$. Then $h(P)$ is the segment $[x^P_{\min},x^P_{\max}]$. The following easy observation will be used in the proof of \Cref{thm:ptasDS}.

\begin{lemma}\label{clm:induced} Let $G$ be a VPG graph with representation $\mathcal{R} = (\mathcal{G},\mathcal{P})$ and such that each path in $\mathcal{P}$ has a number of bends polynomial in $|V(G)|$. Let $H$ be an induced subgraph of $G$. Then we can compute in $O(\vert V(G) \vert)$ time a VPG representation of $H$.
\end{lemma}

\begin{proof} We obtain a VPG representation $\mathcal{R}' = (\mathcal{G},\mathcal{P}')$ of $H$ from $\mathcal{R}$ as follows. Denote by $\mathcal{P}'$ the set of paths whose corresponding vertices belong to $V(H)$. Then, a path $P \in \mathcal{P}'$ is described by the same sequence $s(P)$ as in $\mathcal{R}$ together with the sequence of first intersection points $\{f_{P,P'} : P' \in \mathcal{P}'\setminus P\}$ and the sequence of last intersection points $\{l_{P,P'} : P' \in \mathcal{P}'\setminus P\}$.
\end{proof}

The \textit{refinement of a grid $\mathcal{G}$} having grid-step $\sigma$ is the operation adding a new column (resp. row) between any pair of consecutive columns (resp. rows) in $\mathcal{G}$ and setting the grid-step to $\sigma/2$. Notice that this operation does not change the sequences above.

\section{Hardness results}\label{sechard}

In this section, we show that \textsc{Independent Set} and \textsc{Dominating Set} remain $\mathsf{NP}$-complete when restricted to $B_0$-VPG graphs admitting a representation such that each grid-edge belongs to at most $1$ path and each horizontal path has length at most $2$. In fact, our results hold for a subclass of $B_0$-VPG graphs, namely that of $B_0$-CPG graphs defined in \Cref{prel}. 

Similarly to $B_{k}$-VPG graphs, it was recently shown that recognizing $B_{k}$-CPG graphs is $\mathsf{NP}$-complete, for each fixed $k \geq 0$ \citep{CGMR19,cpg}. Nonetheless, as it will become evident from the proofs, both hardness results (\Cref{thm:boundedlengthIS,thm:boundedlengthDS}) hold even if the $B_{0}$-CPG representation is given as part of the input.

Before turning to the proofs we need the following technical result.

\begin{lemma}\label{maxdeg3} For any subcubic triangle-free $B_0$-CPG graph $G$ on $n$ vertices and with $0$-bend CPG representation $\mathcal{R} = (\mathcal{G},\mathcal{P})$, we can update $\mathcal{R}$ in time polynomial in $n$ so that the following hold:
\begin{itemize}
\item[(a)] $\mathcal{G}$ contains $O(n)$ columns;
\item[(b)] A path $P$ strictly contains one endpoint of another path if and only if the vertex corresponding to $P$ is a $3$-vertex;
\item[(c)] Each horizontal path corresponding to a $2$-vertex has length at least $5$;
\item[(d)] For each horizontal path corresponding to a $3$-vertex $v$, denoting by $p_{\ell}$ and $p_r$ its left and right endpoint, respectively, and by $p$ the contact-point contained in its interior, the segments $p_{\ell}p$ and $pp_r$ have length at least $5$;
\item[(e)] Each path corresponding to a $1$-vertex has length $1$.
\end{itemize}
\end{lemma}

\begin{proof} Let $\mathcal{R} = (\mathcal{G}, \mathcal{P})$ be the input $0$-bend CPG representation. Without loss of generality, the grid-step $\sigma$ is $1$. We begin by preprocessing the grid $\mathcal{G}$ so that it contains $O(n)$ columns. Let $X=(x_1,x_2,\ldots,x_{\ell})$ be the sequence of $x$-coordinates defined as follows:

\begin{itemize}
\item $x_1$ and $x_\ell$ are the smallest and largest $x$-coordinates of paths in $\mathcal{P}$, respectively, that is, there exist $P$ and $P'$ in $\mathcal{P}$ such that $s(P) = (x_1,z),(x,y)$ and $s(P') = (x',y'),(x_\ell,z')$ and, for any $Q \in \mathcal{P}$ with $s(Q) = (a,b),(c,d)$, we have $x_1 \leq \min \{a,c\}$ and $\max \{a,c\} \leq x_\ell$;
\item for any $2 \leq i \leq \ell -1$, $x_i \in X$ if and only if $x_{i} \notin \{x_{1}, x_{\ell}\}$ and there exists a path $Q \in \mathcal{P}$ such that $s(Q) = (x_i,u),(x_i,u')$ for some $u,u' \in \mathbb{R}$;
\item for any $1 \leq i < j \leq \ell$, $x_i < x_j$.
\end{itemize}
For a fixed $1 \leq i \leq \ell-1$, let $y$ be a row of $\mathcal{G}$ such that $\mathcal{P}_y^i = \{Q \in \mathcal{P}: s(Q) = (u,y),(v,y) \text{ with } u < v \text{ and } x_i < v < x_{i+1}\} \cup \{Q \in \mathcal{P}: s(Q) = (u,y),(v,y) \text{ with } x_i < u < x_{i+1} \leq v\}$ is non-empty. Let $X_y^i = ((u_1,u'_1),\ldots,(u_{k_y^i},u'_{k_y^i}))$ be the sequence of $x$-coordinates of endpoints of paths in $\mathcal{P}_y^i$ in increasing order, that is, $(u_p,u'_p) \in X_y^i$ if and only if there exists a path $Q \in \mathcal{P}_y^i$ such that $s(Q) = (u_p,y),(u'_p,y)$ with $u_p < u'_p$, and for any $1 \leq p < q \leq k_y^i$, $u'_p \leq u_q$. We then let $k_i = \max_{y}k_y^i$. Observe that $\sum_{1 \leq i \leq \ell-1} k_i \leq n$. 

Suppose now there exists $1 \leq i \leq \ell -1$ such that $x_{i+1} - x_i > 2k_i + 1$ and consider the smallest such index $i$. The idea is that, since there is no vertical path between columns $x_i$ and $x_{i+1}$, we can ``shrink'' the slice keeping $x_{i}$ fixed so that $x_{i+1} - x_{i} \leq 2k_i + 1$ and repeat. More precisely, for each row $y$ of $\mathcal{G}$ such that $\mathcal{P}_y^i \neq \varnothing$, we proceed as follows (note that there are at most $n$ such rows). Consider the sequence $X_y^i = ((u_1,u'_1),\ldots,(u_{k_y^i},u'_{k_y^i}))$ defined above. If $u_1 > x_i$, we replace the occurrence of $u_1$ in $X_y^i$ with $x_i + 1$, otherwise we replace the at most two occurrences of $u'_1$ in $X_y^i$ with $x_i + 1$. Now suppose that there exists $2 \leq j \leq k_y^i-1$ such that $u'_j > u_j + 1$ and consider the smallest such index $j$. By minimality, $u'_p = u_p + 1$ for any $p < j$. We then replace the at most two occurences of $u'_j$ in $X_y^i$ with $u_j + 1$. Repeating this process, we obtain that $u'_j = u_j + 1$ for any $2 \leq j \leq k_y^i-1$. Finally, if $u'_{k_y^i} < x_{i+1}$, we replace the occurrence of $u'_{k_y^i}$ in $X_y^i$ with $u_{k_y^i} + 1$. Suppose now that there exists $1 \leq j \leq k_y^i -1$ such that $u'_j \neq u_{j+1}$ (and so $u'_j < u_{j+1}$) and consider the smallest such index $j$. For any $j+1 \leq p \leq k_y^i -1$, we replace $(u_p,u'_p)$ in $X_y^i$ with $(u_p - (u_{j+1} - (u'_j + 1)), u'_p - (u_{j+1} - (u'_j + 1)))$. Moreover, if $u'_{k_y^i} < x_{i+1}$, we replace $(u_{k_y^i},u'_{k_y^i})$ in $X_y^i$ with $(u_{k_y^i} - (u_{j+1} - (u'_j + 1)), u'_{k_y^i} - (u_{j+1} - (u'_j + 1)))$, otherwise we replace $(u_{k_y^i},u'_{k_y^i})$ in $X_y^i$ with $(u_{k_y^i} - (u_{j+1} - (u'_j + 1)), u'_{k_y^i})$. Repeating this process, we obtain that $u_{j+1} - u'_j \leq 1$ for any $1 \leq j \leq k_y^i-1$, and so $u_{k_y^i} - x_i \leq 2k_y^i - 1$. For any path $P \in \mathcal{P}$ with endpoints $(x_P^1,y_P^1)$ and $(x_P^2,y_P^2)$ such that $x_P^j = x_{i+1}$, we then replace $x_P^j$ in $s(P)$ with $x_i + 2k_i + 1$. Finally, we update $X$ to $(x_1,\ldots,x_i,x_i + 2k_i +1, x_{i+2}, \ldots,x_{\ell})$, and either for any $1 \leq i \leq \ell-1$, we have $x_{i+1} - x_i \leq 2k_i + 1$, in which case $\mathcal{G}$ uses at most $\sum_{1\leq i \leq \ell -1} (2k_i + 1) + 1 \leq 3n + 1$ columns, or there exists $1 \leq i \leq \ell -1$ such that $x_{i+1} - x_i > 2k_i + 1$, in which case we repeat the procedure above until we obtain a representation satisfying (a).  

Consider now $P \in \mathcal{P}$. Since $P$ has no bend, $\ell_P = 2$. Moreover, paths pairwise touch at most once and so $f_{P,P'} = l_{P,P'}$, for any $P' \neq P$. Suppose first $P$ is such that its first intersection point $f_P$ and last intersection point $l_P$ coincide, that is, $P$ corresponds to a $1$-vertex. If $f_{P} \notin \{(x_1,y_1), (x_{2}, y_{2})\}$, we replace $(x_1,y_1)$ with $f_P$ in $s(P)$. Suppose finally that $P$ is such that $f_P$ and $l_P$ are distinct, that is, $P$ corresponds to a $2^{+}$-vertex. If $(x_1,y_1) \neq f_P$, we replace $(x_1,y_1)$ with $f_P$ in $s(P)$. If $(x_2,y_2) \neq l_P$, we replace $(x_2,y_2)$ with $l_P$ in $s(P)$. Applying this procedure to any path $P \in \mathcal{P}$, the updated representation satisfies (b) and clearly still satisfies (a).

Consider now, in the updated representation, the set $\mathcal{P}_H = \{P \in \mathcal{P}: f_P \neq l_P \text{ and } y_1 = y_2\}$ of horizontal paths whose corresponding vertex has degree at least $2$. Clearly, the length of any path $P \in \mathcal{P}_H$ corresponding to a $2$-vertex is at least $1$. Similarly, if $P \in \mathcal{P}_H$ corresponds to a $3$-vertex and $p$ denotes the contact-point in $P \setminus \partial(P)$, we have that the paths $p_{\ell}p$ and $pp_{r}$ on the grid have both length at least $1$. Therefore, by refining the grid $3$ times, that is, setting the grid-step to $\sigma' = \sigma / 2^3$, any path $P \in \mathcal{P}_H$ corresponding to a $2$-vertex has length at least $8$ and any path $P \in \mathcal{P}_H$ corresponding to a $3$-vertex is such that $p_{\ell}p$ and $pp_{r}$ have length at least $8$. The updated representation satisfies (c) and (d), and still satisfies (a) and (b). 

Consider now a path $P$ in the updated representation such that $f_P = l_P$ , that is, $P$ corresponds to a $1$-vertex. Suppose first that $f_{P} = (x_1,y_1)$. If $P$ is a vertical path (i.e., $x_2 = x_1$), we set $s(P) = (x_1,y_1), (x_1,y_1-\sigma')$ if $y_2 < y_1$, and $s(P) = (x_1,y_1),(x_1,y_1 + \sigma')$ if $y_1 < y_2$. Otherwise, $P$ is a horizontal path (i.e., $y_1 = y_2$) and we set $s(P) = (x_1,y_1), (x_1 - \sigma',y_1)$ if $x_2 < x_1$, and $s(P) = (x_1,y_1),(x_1 + \sigma',y_1)$ if $x_1 < x_2$. We proceed similarly in the case that $f_{P} = (x_2,y_2)$. By repeating this procedure for any such path $P$, the updated representation satisfies (e). It still satisfies (a), (b), (c) and (d). Clearly, any of the above operations can be done in polynomial time.
\end{proof}

\begin{remark}\label{rmkdom} Observe that \Cref{maxdeg3} remains true if we replace $5$ in (c) and (d) with $4$.
\end{remark}

\begin{theorem}
\label{thm:boundedlengthIS}
\textsc{Independent Set} is $\mathsf{NP}$-complete when restricted to $B_0$-CPG graphs admitting a $B_0$-CPG representation where each horizontal path has length at most 2.
\end{theorem}

\begin{proof}
We reduce from \textsc{Independent Set} restricted to triangle-free subcubic $B_0$-CPG graphs which was shown to be $\mathsf{NP}$-hard even if a $B_0$-CPG representation is part of the input \citep{CGMR19}. Given a triangle-free subcubic $B_0$-CPG graph $G$ with $B_0$-CPG representation $\mathcal{R} = (\mathcal{G}, \mathcal{P})$, we construct a graph $G'$ as follows. First, we update $\mathcal{R}$ in polynomial time so that it satisfies (a) to (e) in \Cref{maxdeg3}. We then introduce two operations which will be applied to each $2^{+}$-vertex $v \in V(G)$ corresponding to a horizontal path $P_{v}$. They depend on whether $d_G(v) = 2$ or $d_G(v) = 3$. 

Suppose first that $d_G(v) = 2$ and denote by $l$ the length of $P_v$ (hence, $l \geq 5$). Let $q$ and $r$ be the quotient and remainder, respectively, of the division of $l-5$ by 2. The \textit{$(q,0)$-splitting} of $v$ is the operation replacing $P_v$ with $2q + 4$ horizontal paths $P_1, \ldots , P_{2q + 4}$ (from left to right) of length 1 and one horizontal path $P_{2q+5}$ (at the right extremity) of length $r+1$. 

Suppose now that $d_G(v) = 3$ and let $S_1$, $S_2$ and $S_3$ be the three grid segments obtained by dividing $P_v$ as follows. $S_2$ is the segment strictly containing the contact-point $p$ and with length $2$, $S_1$ is the remaining part of $P_v$ to the left of $S_2$ (hence with length $l_1 \geq 4$) and $S_3$ is the remaining part of $P_v$ to the right of $S_2$ (hence with length $l_3 \geq 4$) (see \Cref{fig:division}).

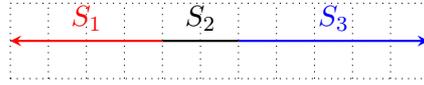
\begin{figure}[htb]
\centering
\begin{tikzpicture}[scale=.5]
\draw[step=1,black,thin,dotted] (0,-1) grid (11,1);
\draw[thick,<-,>=stealth,red] (0,0) -- (4,0) node[midway,above] {$S_1$};
\draw[thick] (4,0) -- (6,0)  node[midway,above] {$S_2$};
\draw[thick,->,>=stealth,blue] (6,0) -- (11,0) node[midway,above] {$S_3$};
\end{tikzpicture}
\caption{Dividing the path $P_v$ into three segments.}
\label{fig:division}
\end{figure} 

For $i \in \{1, 3\}$, let $q_i$ and $r_i$ be the quotient and remainder, respectively, of the division of $l_i - 4$ by 2. The \textit{$(q_1,q_3)$-splitting} of $v$ is the operation replacing: 

\begin{itemize}
\item $S_1$ with $2q_1 + 3$ horizontal paths $P_1, \ldots, P_{2q_1 + 3}$ (from left to right) of length 1 and one horizontal path $P_{2q_1 + 4}$ (at the right extremity) of length $r_1 + 1$;
\item $S_2$ with a horizontal path $P_{2q_1 + 5}$ of length 2; 
\item $S_3$ with $2q_3 + 3$ horizontal paths $P_{2q_1 + 6}, \ldots, P_{2(q_1 + q_3) + 8}$ (from left to right) of length 1 and one horizontal path $P_{2(q_1 + q_3) + 9}$ (at the right extremity) of length $r_3 + 1$. 
\end{itemize}

Notice that the $(q,q')$-splitting of a $2^{+}$-vertex $v$ removes $v$ and replaces it with $2q + 2q' + 4d_G(v) - 3$ new vertices\footnote{We remark that this operation can be obtained by several applications of the \textit{vertex stretching} introduced in \cite{ABKL07} and is in fact equivalent to edge subdivisions.} (if $d_G(v) = 2$, then $q'=0$). The graph $G'$ is then obtained from $G$ by $(q,q')$-splitting every $2^{+}$-vertex $v$ whose corresponding path is horizontal. It is easy to see that this operation can be performed in polynomial time, that $G'$ has $O(n^{2})$ vertices (by \Cref{maxdeg3}) and that it admits a $B_0$-CPG representation where each horizontal path has length at most $2$. To complete the proof, it is then enough to show the following:

\begin{claim}\label{clm:split} Let $H$ be the graph obtained by $(q,q')$-splitting a $2^{+}$-vertex $v\in V(G)$ whose corresponding path in $\mathcal{R}$ is horizontal. We have that $\alpha (H) = \alpha (G) + q + q' + 2(d_G(v) - 1)$.
\end{claim}

Denote by $U = \{v_i : 1 \leq i \leq 2q + 2q' + 4d_G(v) - 3\}$ the set of vertices introduced by the $(q,q')$-splitting of $v$, where $v_iv_{i+1} \in E(H)$ for any $1 \leq i \leq 2q + 2q' + 4d_G(v) - 4$. 

Given a maximum independent set $S$ of $G$, we construct an independent set $S'$ of $H$ as follows. If $v \in S$, then 

$$S' = (S \setminus \{v\}) \cup \{v_{2k+1} : 0 \leq k \leq q + q' + 2(d_G(v) - 1)\},$$ otherwise, $$S' = S \cup \{v_{2k} : 1 \leq k \leq q + q' + 2(d_G(v) - 1)\}.$$ In both cases, $S'$ is easily seen to be independent and so $\alpha (H) \geq \vert S' \vert = \vert S \vert + q + q' + 2(d_G(v) - 1)$.

Conversely, let $S'$ be a maximum independent set of $H$. We construct an independent set $S$ of $G$ as follows. First we add to $S$ every vertex in $S' \setminus U$. Then we decide whether to add $v$ or not according to the following cases. 

Suppose first that $d_G(v) = 2$. If $v_1$ and $v_{2q+5}$ both belong to $S'$, then we add $v$ to $S$ (note that, by maximality,  $S'$ contains $q + 3$ vertices of $U$). Otherwise, one of $v_1$ and $v_{2q+5}$ does not belong to $S'$, in which case we do not add $v$ to $S$ (again, by maximality, $S'$ contains $q + 2$ vertices of $U$). In both cases the constructed $S$ is easily seen to be independent and $\vert S \vert = \vert S' \vert - (q + 2)$. 

Suppose now that $d_G(v) = 3$. If $v_1$, $v_{2q + 5}$ and $v_{2(q+q') + 9}$ belong to $S'$, then we add $v$ to $S$ (note that, by maximality, $S'$ contains $q + q' + 5$ vertices of $U$). Otherwise, one of $v_1$, $v_{2q + 5}$ and $v_{2(q+q') + 9}$ does not belong to $S'$ and we do not add $v$ to $S$ (note that in this case $S'$ contains at most $q + q' + 4$ vertices of $U$). In both cases the constructed $S$ is easily seen to be independent and $\vert S \vert \geq \vert S' \vert - (q + q' + 4)$. 

Therefore, $\alpha (G) \geq \vert S \vert \geq \vert S' \vert - (q + q' + 2(d_G(v) - 1))$, concluding the proof of \Cref{clm:split}. 
\end{proof}

In the following $\mathsf{NP}$-hardness proof, we reduce from \textsc{Dominating Set} restricted to subcubic planar bipartite graphs. This problem is easily seen to be $\mathsf{NP}$-hard by recalling that a $3$-subdivision of an edge of a graph increases the domination number by exactly one \citep{Kor92} and that \textsc{Dominating Set} restricted to subcubic planar graphs is $\mathsf{NP}$-hard \citep{GJ79}.

\begin{theorem}
\label{thm:boundedlengthDS}
\textsc{Dominating Set} is $\mathsf{NP}$-complete when restricted to $B_0$-CPG graphs admitting a $B_0$-CPG representation where each horizontal path has length at most 2.
\end{theorem}

\begin{proof} We reduce from \textsc{Dominating Set} restricted to subcubic planar bipartite graphs which is $\mathsf{NP}$-hard by the paragraph above. Given a subcubic planar bipartite graph $G$, we construct, similarly to \Cref{thm:boundedlengthIS}, a graph $G'$ as follows. First, since $G$ is planar and bipartite, we obtain in linear time a $B_0$-CPG representation $\mathcal{R} = (\mathcal{G}, \mathcal{P})$ of $G$ \citep{CKU98}. Since $G$ is triangle-free and subcubic, we then update $\mathcal{R}$ in polynomial time so that it satisfies (a) to (e) in \Cref{maxdeg3} where the $5$'s in the statements are replaced by $4$'s (see \Cref{rmkdom}). We now introduce two operations which will be applied to each $2^{+}$-vertex $v \in V(G)$ corresponding to a horizontal path $P_{v}$. They depend on whether $d_G(v) = 2$ or $d_G(v) = 3$. 

Suppose first that $d_G(v) = 2$ and denote by $l$ the length of $P_v$ (hence, $l \geq 4$). Let $q$ and $r$ be the quotient and remainder, respectively, of the division of $l-4$ by 3. The \textit{$(q,0)$-splitting} of $v$ is the operation replacing $P_v$ with $3q + 2$ horizontal paths $P_1, \ldots , P_{3q + 2}$ (from left to right) of length 1 and two horizontal paths $P_{3q+3}$ (touching $P_{3q+2}$) and $P_{3q+4}$ (at the right extremity) of lengths $1 + \lfloor r/2 \rfloor$ and $1 + \lceil r/2 \rceil$, respectively. 

Suppose now that $d_G(v) = 3$ and let $S_1$, $S_2$ and $S_3$ be the three grid segments obtained by dividing $P_v$ as follows. $S_2$ is the segment strictly containing the contact-point $p$ and with length $2$, $S_1$ is the remaining part of $P_v$ to the left of $S_2$ (hence with length $l_1 \geq 3$) and $S_3$ is the remaining part of $P_v$ to the right of $S_2$ (hence with length $l_3 \geq 3$). For $i \in \{1, 3\}$, let $q_i$ and $r_i$ be the quotient and remainder, respectively, of the division of $l_i - 3$ by 3. The \textit{$(q_1,q_3)$-splitting} of $v$ is the operation replacing: 

\begin{itemize}
\item $S_1$ with $3q_1 + 1$ horizontal paths $P_1, \ldots, P_{3q_1 + 1}$ (from left to right) of length 1 and two horizontal paths $P_{3q_1+2}$ (touching $P_{3q_1+1}$) and $P_{3q_1 + 3}$ (at the right extremity) of lengths $1 + \lfloor r_1/2 \rfloor$ and $1 + \lceil r_1/2 \rceil$, respectively; 
\item $S_2$ with a horizontal path $P_{3q_1 + 4}$ of length 2; 
\item $S_3$ with $3q_3 + 1$ horizontal paths $P_{3q_1 + 5}, \ldots, P_{3(q_1 + q_3) + 5}$ (from left to right) of length 1 and two horizontal paths $P_{3(q_1+q_3)+6}$ (touching $P_{3(q_1+q_3)+5}$) and $P_{3(q_1 + q_3) + 7}$ (at the right extremity) of lengths $1 + \lfloor r_3/2 \rfloor$ and $1 + \lceil r_3/2 \rceil$, respectively. 
\end{itemize}

Notice that the $(q,q')$-splitting of a $2^{+}$-vertex $v$ removes $v$ and replaces it with $3(q + q' + d_G(v)) - 2$ new vertices (if $d_G(v) = 2$, then $q'=0$). The graph $G'$ is then obtained from $G$ by $(q,q')$-splitting every $2^{+}$-vertex $v$ whose corresponding path is horizontal. It is easy to see that this operation can be performed in polynomial time, that $G'$ has $O(n^{2})$ vertices (by \Cref{maxdeg3}) and that $G'$ admits a $B_0$-CPG representation where each horizontal path has length at most 2. To complete the proof, it is then enough to show the following:

\begin{claim}\label{claimdomm} Let $H$ be the graph obtained by $(q,q')$-splitting a $2^{+}$-vertex $v\in V(G)$ whose corresponding path in $\mathcal{R}$ is horizontal. We have that $\gamma (H) = \gamma (G) + q + q' + d_G(v) - 1$.
\end{claim}

Denote by $U = \{v_i : 1 \leq i \leq 3(q + q' + d_G(v)) - 2\}$ the set of vertices introduced by the $(q,q')$-splitting of $v$, where $v_iv_{i+1} \in E(H)$ for any $1 \leq i \leq 3(q + q' + d_G(v)) - 3$. 

Given a minimum dominating set $D$ of $G$, we construct a dominating set $D'$ of $H$ as follows. If $v \in D$, then $$D' = (D \setminus \{v\}) \cup \{v_{3k + 1} : 0 \leq k \leq q + q' + d_G(v) - 1\}.$$ Otherwise, there exists $u \in N_G(v)$ which belongs to $D$ and we distinguish cases depending on which vertex $u$ is adjacent to in $H$. If $u$ is adjacent to $v_1$, then $$D' = D \cup \{v_{3k} : 1 \leq k \leq q + q' + d_G(v) - 1\}.$$ If $u$ is adjacent to $v_{3(q + q' + d_G(v)) - 2}$, then $$D' = D \cup \{v_{3k + 2} : 0 \leq k \leq q + q' + d_G(v) - 2\}.$$ Otherwise, $u$ is adjacent to $v_{3q + 4}$ and $$D' = D \cup \{v_{3k+2} : 0 \leq k \leq q\} \cup \{v_{3(k+q+1)} : 1 \leq k \leq q' + 1\}.$$ In all cases, $D'$ is easily seen to be dominating and $\gamma (H) \leq \vert D' \vert = \vert D \vert + q + q' + d_G(v) -1$.

Conversely, let $D'$ be a minimum dominating set of $H$. We construct a dominating set $D$ of $G$ as follows. First we put in $D$ every vertex in $D' \setminus U$. Then we decide whether to add $v$ or not according to the following cases. 

Suppose first that $d_G(v) = 2$. Since the vertices $v_{3k}$ with $1 \leq k \leq q+1$ have pairwise disjoint neighborhoods, we have that $|D' \cap U| \geq q+1$. We now claim that if $|D' \cap U| = q+1$, then none of $v_1$ and $v_{3q+4}$ belongs to $D'$ and one of them is dominated by some vertex in $D' \setminus U$. Indeed, $$v_1 \notin \bigcup_{1 \leq k \leq q+1}N[v_{3k}], \ \ \ v_{3q+4} \notin \bigcup_{0 \leq k \leq q}N[v_{3k+2}]$$ and the unions are over pairwise disjoint neighborhoods. Therefore, if $|D' \cap U| = q+1$, then $\{v_1, v_{3q+4}\} \cap D' = \varnothing$. Suppose now that none of $v_1$ and $v_{3q+4}$ is dominated by some vertex in $D' \setminus U$. This implies that $v_2$ and $v_{3q+3}$ both belong to $D'$. But none of $v_2$ and $v_{3q+3}$ belongs to $\bigcup_{1 \leq k \leq q}N[v_{3k+1}]$ and these neighborhoods are pairwise disjoint, contradicting the fact that $|D' \cap U| = q+1$. Therefore, in the case $d_G(v) = 2$, we add $v$ to $D$ if and only if $|D' \cap U| > q + 1$. 

Suppose now that $d_G(v) = 3$. Since the vertices $v_{3k}$ with $1 \leq k \leq q+q'+2$ have pairwise disjoint neighborhoods, we have that $|D' \cap U| \geq q+q'+2$. Similarly to the previous paragraph, we claim that if $|D' \cap U| = q+q'+2$, then none of $v_1$, $v_{3q+4}$ and $v_{3(q+q')+7}$ belongs to $D'$ and one of them is dominated by some vertex in $D' \setminus U$. Indeed, $$v_1 \notin \bigcup_{1 \leq k \leq q+q'+2}N[v_{3k}], \ \ \ v_{3q+4} \notin \bigcup_{0 \leq k \leq q}N[v_{3k+2}] \bigcup_{q + 2 \leq k \leq q + q' + 2}N[v_{3k}], \ \ \  v_{3(q+q')+7}\notin \bigcup_{0 \leq k \leq q+q'+1}N[v_{3k+2}]$$ and the unions are over pairwise disjoint neighborhoods. Therefore, if $|D' \cap U| = q+q'+2$, then $\{v_1, v_{3q+4}, v_{3(q+q')+7}\} \cap D' = \varnothing$. 

Suppose now that none of $v_1$, $v_{3q+4}$ and $v_{3(q+q')+7}$ is dominated by some vertex in $D' \setminus U$. This implies that $v_2$, $v_{3(q+q')+6}$ and one of $\{v_{3q+3}, v_{3q+5}\}$ all belong to $D'$. We assume that $v_{3q+3} \in D'$ (the case $v_{3q+5} \in D'$ is similar and left to the reader). On the other hand, $$\{v_2, v_{3q+3}, v_{3(q+q')+6}\} \cap \bigg(\bigcup_{1 \leq k \leq q}N[v_{3k+1}] \bigcup_{q+1 \leq k \leq q+q'}N[v_{3k+2}]\bigg) = \varnothing$$ and these neighborhoods are pairwise disjoint, contradicting the fact that $|D' \cap U| = q+q'+2$. Therefore, in the case $d_G(v) = 3$, we add $v$ to $D$ if and only if $|D' \cap U| > q + q' + 2$. 

In both cases, $D$ is easily seen to be dominating and so $\gamma (G) \leq \vert D \vert \leq \vert D' \vert - (q + q' + d_G(v) - 1)$, thus concluding the proof of \Cref{claimdomm}.
\end{proof}

%==================================================================

\section{Boundedness of width parameters}\label{mimsec}

In this section, we prove the main result allowing us to obtain PTASes for \textsc{Independent Set} and \textsc{Dominating Set}: If a VPG graph admits a VPG representation with a bounded number of columns and such that each grid-edge belongs to a bounded number of paths, then the graph has bounded mim-width.

The maximum induced matching width (mim-width for short) is a graph parameter introduced by \citet{Vat12}
measuring how easy it is to decompose a graph along vertex cuts inducing a bipartite graph with small maximum induced matching size. Replacing induced matchings with matchings, one obtains the related parameter called maximum matching width (mm-width for short) \citep{Vat12}. The modelling power of mim-width is stronger than that of tree-width, in the sense that graphs of bounded tree-width have bounded mim-width but there exist graph classes (interval graphs and permutation graphs) with mim-width $1$ \citep{Vat12} and unbounded tree-width \citep{GR00}. On the other hand, mm-width and tree-width are equivalent parameters, in the sense that one is bounded if and only if the other is \citep{Vat12}.

It is well-known that boundedness of tree-width allows polynomial-time solvability of several otherwise $\mathsf{NP}$-hard graph problems. Boundedness of mim-width has important algorithmic consequences as well, in particular for the so-called $(\sigma, \rho)$-domination problems, a subclass of graph problems expressible in $\mbox{MSO}_{1}$ introduced by \citet{TP97} and including \textsc{Independent Set} and \textsc{Dominating Set}. Combining results in \citep{BV13,BTV13}, it is known that the three versions of a $(\sigma, \rho)$-domination problem (minimization, maximization, existence) can be solved in $O(n^{w})$ time, assuming a branch decomposition of mim-width $w$ is provided as part of the input. In the case of \textsc{Independent Set} and \textsc{Dominating Set}, these results read as follows: 

\begin{theorem}[see \citep{JKST18}]\label{algomim} There is an algorithm that, given a graph $G$ and a branch decomposition $(T, \delta)$ of $G$ with $w = \mathrm{mimw}_{G}(T, \delta)$, solves \textsc{Independent Set} and \textsc{Dominating Set} in $O(n^{4 + 3w})$ time.
\end{theorem}

It should be remarked that deciding the mim-width of a graph is $\mathsf{NP}$-hard in general and not in $\mathsf{APX}$ unless $\mathsf{NP} = \mathsf{ZPP}$ \citep{SV16}. However, \citet{BV13} showed that it is possible to find branch decompositions of constant mim-width in polynomial time for several classes of graphs such as permutation graphs, convex graphs, interval graphs, circular arc graphs, etc.

\citet{JKST18} enlarged the class of problems polynomially solvable on graphs of bounded mim-width by showing that the distance-$r$ version of a $(\sigma, \rho)$-domination problem can be polynomially reduced to the $(\sigma, \rho)$-domination problem. This essentially follows from the fact that, for each positive integer $r$, the mim-width of the graph power $G^{r}$ is at most twice that of $G$. The effect on mim-width of some other graph operations has been studied in \citep{BHM20,GMR20,Men17}.

Before turning to the proof of our result, let us properly define the notions of mim-width and mm-width. A \textit{branch decomposition}\footnote{A branch decomposition is also known as a decomposition tree.} for a graph $G$ is a pair $(T, \delta)$, where $T$ is a subcubic tree and $\delta$ is a bijection between the vertices of $G$ and the leaves of $T$. Each edge $e \in E(T)$ naturally splits the leaves of the tree in two groups depending on their component when $e$ is removed. In this way, each edge $e \in E(T)$ represents a partition of $V(G)$ into two partition classes $A_{e}$ and $\overline{A_{e}}$, denoted $(A_{e}, \overline{A_{e}})$. Denoting by $G[X, Y]$ the bipartite subgraph of $G$ induced by the edges with one endpoint in $X$ and the other in $Y$, mim-width and mm-width are defined as follows:

\begin{definition} Let $G$ be a graph and let $(T, \delta)$ be a branch decomposition for $G$. For each edge $e \in E(T)$ and the corresponding partition $(A_{e}, \overline{A_{e}})$ of $V(G)$, we denote by $\mathrm{cutmim}_{G}(A_{e}, \overline{A_{e}})$ and $\mathrm{cutmm}_{G}(A_{e}, \overline{A_{e}})$ the size of a maximum induced matching and maximum matching in $G[A_{e}, \overline{A_{e}}]$, respectively. The mim-width of the branch decomposition $(T, \delta)$ is the quantity $\mathrm{mimw}_{G}(T, \delta) = \max_{e \in E(T)}\mathrm{cutmim}_{G}(A_{e}, \overline{A_{e}})$. The mim-width $\mathrm{mimw}(G)$ of the graph $G$ is the minimum value of $\mathrm{mimw}_{G}(T, \delta)$ over all possible decompositions trees $(T, \delta)$ for $G$. The mm-width of the branch decomposition $(T, \delta)$ is the quantity $\mathrm{mmw}_{G}(T, \delta) = \max_{e \in E(T)}\mathrm{cutmm}_{G}(A_{e}, \overline{A_{e}})$ and the mm-width $\mathrm{mmw}(G)$ of the graph $G$ is the minimum value of $\mathrm{mmw}_{G}(T, \delta)$ over all possible decompositions trees $(T, \delta)$ for $G$. 
\end{definition}

Clearly, for any branch decomposition $(T, \delta)$ for $G$, $\mathrm{mimw}_G(T, \delta) \leq \mathrm{mmw}_G(T, \delta)$ and so $\mathrm{mimw}(G) \leq \mathrm{mmw}(G)$. 

We assume throughout the rest of the paper that our graphs have no isolated vertices. This can be safely assumed as input graphs can be easily preprocessed in order to remove isolated vertices when we consider \textsc{Independent Set} and \textsc{Dominating Set}. Notice however that our following key result still holds if the graph has isolated vertices. It shows that the graph class of interest not only has bounded mim-width but also bounded mm-width.

\begin{theorem}\label{lem:boundedmimvpg}
Let $G$ be a VPG graph with a representation $\mathcal{R} = (\mathcal{G}, \mathcal{P})$ such that each grid-edge in $\mathcal{G}$ belongs to at most $t$ paths in $\mathcal{P}$ and $\mathcal{G}$ contains at most $\ell$ columns, for some integers $t,\ell \geq 0$. Then $\mathrm{mimw}(G) \leq \mathrm{mmw}(G) \leq 3t \cdot (\ell + 1)$. 

Moreover, if we are given a VPG graph $G$ on $n$ vertices together with a representation $\mathcal{R} = (\mathcal{G}, \mathcal{P})$ as above and such that in addition each path in $\mathcal{P}$ has a number of bends polynomial in $n$, then it is possible to compute in $O(n\log n)$ time a branch decomposition $(T, \delta)$ for $G$ such that $\mathrm{mimw}_G(T, \delta) \leq\mathrm{mmw}_G(T, \delta) \leq 3t \cdot (\ell + 1)$.   
\end{theorem}

\begin{proof} It is enough to show only the second assertion as it will become clear from the proof that the constraint on the number of bends is used only to efficiently compute $(T, \delta)$ (recall that in this case our data structure has polynomial size). 

We begin by modifying $\mathcal{R}$ in $O(n)$ time so that, for each path $P \in \mathcal{P}$, both of its endpoints belong to at least one other path, unless $P$ contains only one intersection point, in which case at least one of its endpoints belongs to another path. This is done as follows. If there exists $P \in \mathcal{P}$ such that $(x_1,y_1)$ belongs to no other path, we replace $(x_1,y_1)$ with the first intersection point $f_P$ of $P$ in $s(P)$. Similarly, if $(x_{\ell_P},y_{\ell_P})$ belongs to no other path, we replace $(x_{\ell_P},y_{\ell_P})$ with the last intersection point $l_P$ of $P$ in $s(P)$, unless $f_P = l_P$. 
%By possibly shortening paths, we first modify $\mathcal{R}$ in $O(n)$ time so that, for each path $P \in \mathcal{P}$, both of its endpoints belong to at least one other path, unless $P$ contains only one intersection point, in which case only one of its endpoints belongs to another path. 

We now construct a branch decomposition $(T, \delta)$ for $G$ as follows. Let $(p_i)_{1 \leq i \leq k}$ be the sequence of intersection points between paths in $\mathcal{P}$ ordered starting from the upper-most left-most intersection point and ending on the lower-most right-most intersection point, by reading each line from left to right. More precisely, $p_i = (x_i, y_i) < (x_j, y_j) = p_j$ if $y_i > y_j$, or $y_i = y_j$ and $x_i < x_j$. We say that two paths $P$ and $P'$ are \textit{equivalent}, denoted $P \equiv P'$, if $\min \{i \in \{1, \dots ,k\}: p_i \in \partial (P)\} = \min \{i \in \{1, \dots ,k\}: p_i \in \partial (P')\} = j$ and $P$ and $P'$ use the same grid-edge incident to $p_j$. It is easy to see that $\equiv$ is an equivalence relation on the set of paths in $\mathcal{P}$. We then define a total order $\preceq$ on the set of equivalence classes as follows. For any two distinct equivalence classes $[P]$ and $[P']$, $[P] \prec [P']$ if $\min \{i \in \{1, \dots ,k\}: p_i \in \partial (P)\} < \min \{i \in \{1, \dots ,k\}: p_i \in \partial (P')\}$ or $\min \{i \in \{1, \ldots ,k\}: p_i \in \partial (P)\} = \min \{i \in \{1, \dots ,k\}: p_i \in \partial (P')\}$, in which case we break the tie as shown in \Cref{fig:tie} (the equivalence classes are ordered in clockwise order, where the smallest equivalence class is the one containing paths that use the grid-edge to the left of the intersection point).

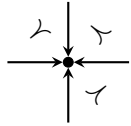
\begin{figure}[htb]
\centering
\begin{tikzpicture}[scale=.8]
\node[circ] (p) at (0,0) {};
\draw[->,>=stealth,thick] (-1,0) -- (p);
\draw[->,>=stealth,thick] (1,0) -- (p);
\draw[->,>=stealth,thick] (0,-1) -- (p);
\draw[->,>=stealth,thick] (0,1) -- (p);

\node[draw=none,rotate=40] at (-.5,.5) {\large $\prec$};
\node[draw=none,rotate=-40] at (.5,.5) {\large $\prec$};
\node[draw=none,rotate=220] at (.5,-.5) {\large $\prec$};
\end{tikzpicture}
\caption{How to break ties.}
\label{fig:tie}
\end{figure}

Consider now a caterpillar $T$ built from a path $v_{1}v_{2}\cdots v_{n}$ by attaching a pendant vertex $w_{i}$ to each internal vertex $v_{i}$ of the path. Letting $w_{1} = v_{1}$ and $w_{n} = v_{n}$, we have that $w_1, \dots, w_{n}$ are the leaves of $T$. We then map the vertices of $G$ to the leaves of $T$ so that the order $\preceq$ is preserved and by arbitrarily breaking the ties within each equivalence class. In particular, if $u,v \in V(G)$ are such that $[P_u] \prec [P_v]$, then $\delta (u) = w_i$ and $\delta (v) = w_j$ for some $i < j$. Note that the branch decomposition $(T, \delta)$ can be built in $O(n\log n)$ time by sorting the at most $2n$ intersection points which are endpoints of paths in $\mathcal{P}$. It remains to show that $\mathrm{mmw}_G(T, \delta) \leq 3t \cdot (\ell + 1)$.   

Let $e \in E(T)$. Clearly, we may assume that $e$ is not incident to a leaf and so, for some $1< s < n-1$, we have that $A_e = \{ u \in V(G): \delta (u) = w_i \text{ with } 1 \leq i \leq s\}$ and $\overline{A_e} = \{ u \in V(G): \delta (u) = w_i \text{ with } s < i \leq n\}$ are the corresponding partition classes of $V(G)$. Consider a vertex $u \in A_e$ such that, for any $v \in A_e$, either $P_v \equiv P_u$ or $[P_v] \prec [P_u]$. Let $p$ be the smallest intersection point in $\partial (P_u)$ i.e., $p = p_j$ where $j = \min \{i \in \{1, \dots ,k\}: p_i \in \partial (P_u)\}$. The grid-point $p$ naturally divides the grid $\mathcal{G}$ into two parts: the \textit{upper part}, containing $p$ and every intersection point smaller than $p$, and the \textit{lower part}, containing every intersection point larger than $p$. We denote by $\mathcal{L}$ the set of grid-edges in the upper part bordering the lower part (the red line in \Cref{fig:border}) and by $\mathcal{C}$ the set of grid-points in $\mathcal{L}$ (the grid-points on the red line in \Cref{fig:border}). Clearly, $\vert \mathcal{C} \vert \leq \ell + 1$.

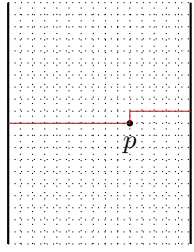
\begin{figure}[htb]
\centering
\begin{tikzpicture}[scale=.8]
\draw[step=.2,black,thin,dotted] (0,0) grid (3,4);
\draw[thick] (0,0) -- (0,4);
\draw[thick] (3,0) -- (3,4);

\node[cir,label=below:{\small $p$}] at (2,2) {};
\draw[red] (0,2) -- (2,2) -- (2,2.2) -- (3,2.2);
\end{tikzpicture}
\caption{The division of $\mathcal{G}$ induced by $p$.}
\label{fig:border}
\end{figure} 

Let $M$ be a maximum matching in $G[A_e,\overline{A_e}]$. Let $\mathcal{P}_e$ and $\overline{\mathcal{P}_e}$ be the sets of paths whose corresponding vertices belong to $A_e$ and $\overline{A_e}$, respectively, and which are matched in $M$. For any $P \in \mathcal{P}_e$, we denote by $\overline{P}$ the path in $\overline{\mathcal{P}_e}$ such that the corresponding vertex is matched in $M$ to the vertex corresponding to $P$. Consider now $P \in \mathcal{P}_e$ and $\overline{P} \in \overline{\mathcal{P}_e}$. They intersect in some grid-point $p'$. We claim that one of $P$ and $\overline{P}$ contains a point of $\mathcal{C}$. Clearly, me may assume that $p'\neq p$ and that $p$ is not an endpoint of $\overline{P}$. By construction, $P \in \mathcal{P}_e$ has at least one endpoint in the upper part. Similarly, $\overline{P}$ has at least one endpoint in the lower part, or else $\overline{P}$ has an endpoint $q$ belonging to another path (i.e., $q = p_{i}$ for some $1\leq i\leq k$) such that $q < p$ and so $\overline{P} = P_{v}$ for some $v \in \overline{A_{e}}$ with $[P_v] \prec [P_u]$, a contradiction. This implies that if $p < p'$, then $P$ contains a point of $\mathcal{C}$, and if $p' < p$, then $\overline{P}$ contains a point of $\mathcal{C}$, as claimed.

For any $P \in \mathcal{P}_e$, denote by $c_{P,\overline{P}} \in \mathcal{C}$ the left-most lower-most grid-point among those grid-points in $\mathcal{C}$ belonging to $P \cup \overline{P}$. For any $q \in \mathcal{C}$, let $S_{q} = \{P \in \mathcal{P}_e : c_{P,\overline{P}} = q  \}$. By the paragraph above, each $P \in \mathcal{P}_e$ is contained in one $S_{q}$. Consider now $q \in \mathcal{C} \setminus \{p_{0}\}$, where $p_{0}$ is the grid-point above $p$. For any $P \in S_{q}$, we have that $c_{P,\overline{P}} = q$ and, by construction, the grid-edge to the left of $q$ does not belong to $P \cup \overline{P}$. But since each grid-edge belongs to at most $t$ paths, $\vert S_q \vert \leq 3t$. Similarly, for any $P \in S_{p_{0}}$, the grid-edge below $p_{0}$ does not belong to $P \cup \overline{P}$ and so $\vert S_{p_{0}} \vert \leq 3t$. It follows that $\mathrm{cutmm}_{G}(A_{e}, \overline{A_{e}}) = \vert \mathcal{P}_e \vert = \sum_{q \in \mathcal{C}}|S_{q}| \leq 3t \cdot \vert \mathcal{C} \vert \leq 3t \cdot (\ell + 1)$.
\end{proof}

Recalling that for any graph $G$ with tree-width $\mathrm{tw}(G)$ we have $\mathrm{tw}(G) \leq 3\cdot\mathrm{mmw}(G) - 1$ \citep{Vat12}, the following holds:

\begin{corollary}\label{corVPG} Let $G$ be a VPG graph with a representation $\mathcal{R} = (\mathcal{G}, \mathcal{P})$ such that each grid-edge in $\mathcal{G}$ belongs to at most $t$ paths in $\mathcal{P}$ and $\mathcal{G}$ contains at most $\ell$ columns, for some integers $t,\ell \geq 0$. Then $\mathrm{tw}(G) \leq 9t \cdot (\ell + 1) - 1$. 
\end{corollary}

We remark that, in bounding tree-width for the graphs in \Cref{corVPG}, it seems more natural to work with branch decompositions rather than directly use tree decompositions as in the classical definition of tree-width. In fact, it appears not to be known whether tree-width admits a characterization in terms of branch decompositions \citep{Vat12}. On the other hand, the bound above is unlikely to be tight and it would be interesting to provide a direct proof. 

\begin{remark}\label{bestposs} We now observe that \Cref{lem:boundedmimvpg} is best possible in the following strong sense: both conditions on the VPG graph are necessary to guarantee boundedness of mim-width. As for the first, consider grid graphs. Since any grid graph is planar and bipartite, it admits a $B_0$-CPG representation $(\mathcal{G}, \mathcal{P})$ \citep{CKU98} and so each grid-edge in $\mathcal{G}$ belongs to at most one path in $\mathcal{P}$. On the other hand, grid graphs do not have bounded mim-width \citep{Vat12}. 

As for the second, consider split graphs. They admit a VPG representation $(\mathcal{G}, \mathcal{P})$ such that $\mathcal{G}$ contains at most $4$ columns (see \Cref{splitrep}) but they do not have bounded mim-width \citep{Men17}. 
\end{remark}

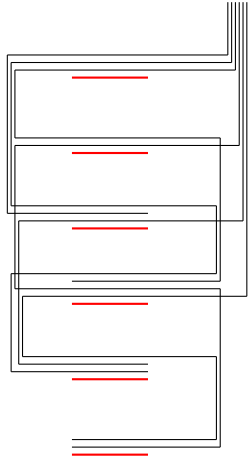
\begin{figure}
\centering
\begin{tikzpicture}
\draw[thick,red] (0,0) -- (1,0) node[midway,below] {\footnotesize }; %12
\draw[thick,red] (0,1) -- (1,1) node[midway,below] {\footnotesize }; %10 
\draw[thick,red] (0,2) -- (1,2) node[midway,below] {\footnotesize }; %7
\draw[thick,red] (0,3) -- (1,3) node[midway,below] {\footnotesize }; %5
\draw[thick,red] (0,4) -- (1,4) node[midway,below] {\footnotesize }; %4
\draw[thick,red] (0,5) -- (1,5) node[midway,below] {\footnotesize }; %2

\draw (2.05,5.3) -- (2.05,6);
\draw (-.85,5.3) -- (2.05,5.3) node[pos=.76,above] {\footnotesize };
\draw (-.85,5.3) -- (-.85,3.2);
\draw (-.85,3.2) -- (1,3.2); %1

\draw (2.15,5.1)-- (2.15,6);
\draw (-.75,5.1) -- (2.15,5.1) node[pos=.73,below] {\footnotesize };
\draw (-.75,5.1) -- (-.75,4.2);
\draw (-.75,4.2) -- (1.95,4.2);
\draw (1.95,4.2) -- (1.95,2.3); 
\draw (0,2.3) -- (1.95,2.3); %3

\draw (2.1,5.2) -- (2.1,6);
\draw (-.8,5.2) -- (2.1,5.2); 
\draw (-.8,5.2) -- (-.8,3.3);
\draw (-.8,3.3) -- (1.9,3.3);
\draw (1.9,3.3) -- (1.9,2.4); 
\draw (-.8,2.4) -- (1.9,2.4);
\draw (-.8,2.4) -- (-.8,1.1) node[midway,left] {\footnotesize }; 
\draw (-.8,1.1) -- (1,1.1); %6

\draw (2.2,4.1) -- (2.2,6);
\draw (-.75,4.1) -- (2.2,4.1) node[pos=.72,below] {\footnotesize };
\draw (-.75,4.1) -- (-.75,2.2);
\draw (-.75,2.2) -- (1.95,2.2);
\draw (1.95,2.2) -- (1.95,.1);
\draw (0,.1) -- (1.95,.1); %8

\draw (2.25,3.1) -- (2.25,6);
\draw (-.7,3.1) -- (2.25,3.1) node[pos=.71,below] {\footnotesize };
\draw  (-.7,3.1) --  (-.7,1.2);
\draw (-.7,1.2) -- (1,1.2); %9

\draw (2.3,2.1) -- (2.3,6);
\draw (-.65,2.1) -- (2.3,2.1) node[pos=.7,below] {\footnotesize };
\draw (-.65,2.1) -- (-.65,1.3);
\draw (-.65,1.3) -- (1.9,1.3);
\draw (1.9,1.3) -- (1.9,.2);
\draw (0,.2) -- (1.9,.2); %11 
\end{tikzpicture}
\caption{A VPG representation on a bounded number of columns of a split graph whose vertices are partitioned into an independent set $I$ and a clique $C$. Red paths correspond to vertices in $I$ and black paths correspond to vertices in $C$.}\label{splitrep}
\end{figure}

\begin{corollary}\label{algogrid} There is an algorithm that, given a VPG graph on $n$ vertices with a representation $\mathcal{R} = (\mathcal{G}, \mathcal{P})$ such that $\mathcal{G}$ contains at most $\ell$ columns and each grid-edge in $\mathcal{G}$ belongs to at most $t$ paths in $\mathcal{P}$, for some integers $\ell,t \geq 0$, solves \textsc{Independent Set} and \textsc{Dominating Set} in $O(n^{4 + 9t(\ell + 1)})$ time.
\end{corollary}

\begin{proof} By \Cref{lem:boundedmimvpg}, we compute in $O(n\log n)$ time a branch decomposition $(T, \delta)$ for $G$ such that $\mathrm{mimw}_G(T, \delta) \leq 3t \cdot (\ell + 1)$. The result then follows from \Cref{algomim}.
\end{proof}

The following lemma will be used together with \Cref{lem:boundedmimvpg} in the PTAS for \textsc{Dominating Set}.

\begin{lemma}
\label{lem:cliquemimwidth}
Let $G=(V,E)$ be a graph and let $S \subseteq V$. Let $G'=(V',E')$ denotes the graph with $V'=V$ and $E' = E \cup \{uv : u,v \in S\}$. If $(T, \delta)$ is a branch decomposition for $G$, then it is also a branch decomposition for $G'$ and $\mathrm{mimw}_{G'}(T, \delta) \leq \mathrm{mimw}_G(T, \delta) + 1$.
\end{lemma}

\begin{proof}
Let $(T,\delta)$ be a branch decomposition for $G$. Since $G$ and $G'$ have the same vertex set, $(T,\delta)$ is a branch decomposition for $G'$ as well. Consider now an edge $e \in E(T)$ and the corresponding partition $(A_e,\overline{A_e})$ such that $\mathrm{cutmim}_{G'}(A_e,\overline{A_e})$ attains the maximum over all edges of $T$ i.e., $\mathrm{mimw}_{G'}(T,\delta) = \mathrm{cutmim}_{G'}(A_e,\overline{A_e})$, and let $M'$ be a maximum induced matching in $G'[A_e,\overline{A_e}]$. Note that $M'$ contains at most one of the edges in $E' \setminus E$ and, by possibly removing this edge, we obtain an induced matching $M$ in $G[A_e,\overline{A_e}]$. Therefore, $\mathrm{mimw}_{G'}(T,\delta) = |M'| \leq |M| + 1 \leq \mathrm{cutmim}_G (A_e,\overline{A_e}) + 1 \leq \mathrm{mimw}_G(T,\delta) + 1$.
\end{proof}

\section{PTASes}\label{ptases}

Combining the machinery developed in \Cref{mimsec} and the well-known Baker's technique, we can finally provide our PTASes for \textsc{Independent Set} and \textsc{Dominating Set}.  

\begin{theorem}
\label{thm:ptasIS} Let $t \geq 0$ and $c \geq 1$ be integers. \textsc{Independent Set} admits a PTAS when restricted to VPG graphs with a representation $\mathcal{R} = (\mathcal{G}, \mathcal{P})$ such that: 
\begin{enumerate}
\item each path in $\mathcal{P}$ has a polynomial (in $\vert \mathcal{P} \vert$) number of bends;
\item\label{2nd} each grid-edge in $\mathcal{G}$ belongs to at most $t$ paths in $\mathcal{P}$;
\item\label{3rd} the horizontal part of each path in $\mathcal{P}$ has length at most $c$.
\end{enumerate} 
\end{theorem}

\begin{proof}
Let $G$ be a VPG graph on $n$ vertices with a representation $\mathcal{R} = (\mathcal{G}, \mathcal{P})$ satisfying the three conditions above. Without loss of generality, we may assume that all the paths in $\mathcal{P}$ contain only grid-points with non-negative coordinates. Moreover, we may assume that $G$ is connected. Therefore, no column in $\mathcal{G}$ is unused and so $\mathcal{G}$ has at most $(c+1)n$ columns. Further note that since any path $P \in \mathcal{P}$ has a polynomial (in $n$) number of bends, the sequence $s(P)$ has polynomial size and we can compute the horizontal part $h(P) = [x^P_{\min},x^P_{\max}]$ of $P$ in polynomial time. Given $0 < \varepsilon < 1$, we fix $k = \lceil 1/\varepsilon \rceil$. 

For any $i \in \mathbb{N}$, we denote by $X_i$ the set of vertices whose corresponding path contains a grid-edge $[(i,j),(i+1,j)]$ for some $j \in \mathbb{N}$. Notice that $X_{i} = \{v \in V(G) : x^{P_{v}}_{\min} \leq i < i+1 \leq x^{P_{v}}_{\max}\}$ and so we can compute the at most $(c+1)n -1$ non-empty sets $X_{i}'s$ in polynomial time. For any $d \in \{0, \ldots, kc-1\}$, let $V_d = \bigcup_{\ell \in \mathbb{N}} X_{d + \ell kc}$ be the set of vertices whose corresponding path contains a grid-edge $[(d +\ell kc,j),(d+\ell kc +1,j)]$ for some $\ell,j \in \mathbb{N}$. We now claim that, for any $d \in \{0,\ldots,kc-1\}$, $G - V_d$ is disconnected. Indeed, after deleting $V_d$, no vertex whose horizontal part is contained in the interval $[0, d + \ell kc]$ can be adjacent to a vertex whose horizontal part is contained in the interval $[d + \ell kc+1, (c+1)n]$. Similarly, every component of $G - V_d$ admits a VPG representation in which the number of columns is bounded by $kc$. By \Cref{algogrid}, for each component of $G - V_d$, we compute a maximum-size independent set in $O(n^{4 + 9t(kc+1)})$ time. The union $U_d$ of these independent sets over the components of $G - V_d$ is then an independent set of $G$ and, after repeating the procedure above for each $d \in \{0,\ldots, kc-1\}$, we return the largest set $U$ among the $U_d$'s.

It remains to show that $|U| \geq (1 - \varepsilon)|\mathsf{OPT}|$, where $\mathsf{OPT}$ denotes an optimal solution of \textsc{Independent Set} with instance $G$. Note that, for any $d \in \{0,\ldots , kc-1\}$, $\mathsf{OPT} \cap V_d$ is the set of vertices in $\mathsf{OPT}$ whose corresponding path contains a grid-edge $[(d +\ell kc,j),(d+\ell kc +1,j)]$ for some $\ell,j \in \mathbb{N}$. Since the horizontal part of each path has length at most $c$, we have that every vertex in $\mathsf{OPT}$ belongs to at most $c$ distinct $V_d$'s. Therefore, denoting by  $d_0$ the index attaining $\min_{d \in \{0, \ldots, kc-1\}} \vert \mathsf{OPT} \cap V_d \vert$, we have $$kc \vert \mathsf{OPT} \cap V_{d_0} \vert \leq \sum_{d=0}^{kc-1} \vert \mathsf{OPT} \cap V_d \vert \leq c \vert \mathsf{OPT} \vert$$ and so $$|\mathsf{OPT}| = |\mathsf{OPT}\setminus V_{d_0}| + |\mathsf{OPT} \cap V_{d_0}| \leq |U| + \varepsilon|\mathsf{OPT}|,$$ thus concluding the proof. 
\end{proof}

\begin{theorem}
\label{thm:ptasDS}
Let $t \geq 0$ and $c \geq 1$ be integers. \textsc{Dominating Set} admits a PTAS when restricted to VPG graphs with a representation $\mathcal{R} = (\mathcal{G}, \mathcal{P})$ such that:
\begin{enumerate}
\item each path in $\mathcal{P}$ has a polynomial (in $\vert \mathcal{P} \vert$) number of bends;
\item\label{2dom} each grid-edge in $\mathcal{G}$ belongs to at most $t$ paths in $\mathcal{P}$;
\item\label{3dom} the horizontal part of each path in $\mathcal{P}$ has length at most $c$.
\end{enumerate}
\end{theorem}

\begin{proof}
Let $G$ be a VPG graph on $n$ vertices with a representation $\mathcal{R} = (\mathcal{G}, \mathcal{P})$ satisfying the three conditions above. We may assume that $G$ is connected. This implies that no column in $\mathcal{G}$ is unused and so the number of columns $m$ is at most $(c+1)n$. Without loss of generality, the paths in $\mathcal{P}$ contain only grid-points with $x$-coordinates between $0$ and $m-1$ and $y$-coordinates at least $0$. As in \Cref{thm:ptasIS}, we compute in time  polynomial in $n$ the horizontal part $h(P) = [x^P_{\min},x^P_{\max}]$ of any $P \in \mathcal{P}$. Given $0 < \varepsilon < 1$, we fix $k = \lceil c(\frac{2}{\varepsilon}-1)\rceil$. Clearly, $k > c$. We finally assume that $m > k + 2c - 1$, for otherwise we can compute an exact solution by \Cref{algogrid}. We then proceed as follows.

For a fixed $s \in \{0,\ldots, k + c -1\}$, we let $r = \lceil \frac{m-s-c}{k + c} \rceil$. Notice that $r = O(n)$. For any $i \in \{0, \ldots, m-2\}$, we compute in polynomial time the set $X_i$ of vertices whose corresponding path contains a grid-edge $[(i,j),(i+1,j)]$ for some $j \in \mathbb{N}$ (we also set $X_{-1} = \varnothing$ and $X_{m-1} = \varnothing$). For any $j \in \{0,\ldots,m-1\}$, we compute in polynomial time the set $V_j$ of vertices whose corresponding path intersects column $j$ (clearly, $V_j = \{v \in V(G) : x^{P_{v}}_{\min} \leq j \leq x^{P_{v}}_{\max}\}$). Note that if $i \in \{0, \ldots, m-2\}$, then $X_i \subset V_i$ and $X_i \subset V_{i+1}$. For any $i \in \{0, \dots, r\}$, we now define $V(i, s)$ and the set $E(i, s)$ of \textit{exterior vertices} of $V(i,s)$ as follows (see \Cref{ext}):  

\begin{itemize}
\item $V(0,s) = \bigcup_{\ell = 0}^{s+c-1} V_\ell$ and $E(0,s) = X_{s+c-1}$;  
\item For $0 < i < r$, $V(i,s) = \bigcup_{\ell = 0}^{k + 2c -1} V_{(i-1) \cdot (k + c) + s + \ell}$ and $E(i, s)$ is the union of $E_L(i,s) = X_{(i-1) \cdot (k + c) + s -1}$ and $E_R(i,s) = X_{i \cdot (k + c) + s + c -1}$;
\item $V(r,s) = \bigcup_{\ell = (r-1) \cdot (k + c) + s}^{m-1} V_\ell$ and $E(r,s) = X_{(r-1) \cdot (k + c) + s-1}$.    
\end{itemize}
Moreover, for any $i \in \{0, \dots, r\}$, the set $I(i, s)$ of \textit{interior vertices} of $V(i,s)$ is defined as $V(i, s) \setminus E(i, s)$. Since the horizontal part of each path in $\mathcal{P}$ has length at most $c$ and recalling that $k > c$, it is not difficult to see that the following holds.

\begin{observation}
\label{obs:Vis}
If $|i - j| > 1$, $V(i,s) \cap V(j,s) = \varnothing$. Moreover, $V(i,s) \cap V(i+1,s) = \bigcup_{\ell = 0}^{c-1} V_{i \cdot (k + c) + s + \ell}$.
\end{observation}

\begin{figure}[htb]

\centering

\begin{tikzpicture}

\draw (.2,5) edge[bend left=10] (3,5);

\node[draw=none] at (1.6,5.4) {\footnotesize $V(i,s)$};

\draw[fill=lightgray,draw=none] (0,0) rectangle (.2,5);

\draw[thick] (0,.1) -- (0,-.1) node[left] {\scriptsize $(i-1)(k + c) + s -1$};

\draw (.2,0) -- (.2,.05);

\draw (.4,0) -- (.4,.05);

\draw[fill=lightgray,draw=none] (.6,0) rectangle (.8,5);

\draw[thick] (.6,.1) -- (.6,-.1);

\node[draw=none] at (1.6,-.3) {\scriptsize $(i-1)(k+c) + s + c-1$};

\draw[thick] (-.2,1.5) -- (.4,1.5) node[midway,below] {\scriptsize $\in E_L(i,s)$} -- (.4,2);

\draw (.8,0) -- (.8,.05);

\draw (1,0) -- (1,.05);

\draw (1.2,0) -- (1.2,.05);

\draw (1.4,0) -- (1.4,.05);

\draw (1.6,0) -- (1.6,.05);

\draw (1.8,0) -- (1.8,.05);

\draw (2,0) -- (2,.05);

\draw (2.2,0) -- (2.2,.05);

\draw (2.4,0) -- (2.4,.05);

\draw (2.6,0) -- (2.6,.05);

\draw (2.8,0) -- (2.8,.05);

\draw (-2.2,5) edge[bend left=10] (0.6,5);

\node[draw=none] at (-.8,5.4) {\footnotesize $V(i-1,s)$};

\draw[fill=lightgray,draw=none] (3,0) rectangle (3.2,5);

\draw (3,0) -- (3,.05);

\draw[thick] (3.2,.1) -- (3.2,-.1) node[right] {\scriptsize $i(k + c) + s + c$};

\draw[thick] (2.8,4) -- (3.4,4) node[midway,above] {\scriptsize $\in E_R(i,s)$} -- (3.4,3);

\draw (0,0) -- (3.2,0);

\draw[thick] (.4,3) -- (1,3) -- (1,2.5) node[midway,right] {\scriptsize $\in E_R(i-1,s)$};

\end{tikzpicture}

\caption{The exterior vertices.}\label{ext}
\end{figure}
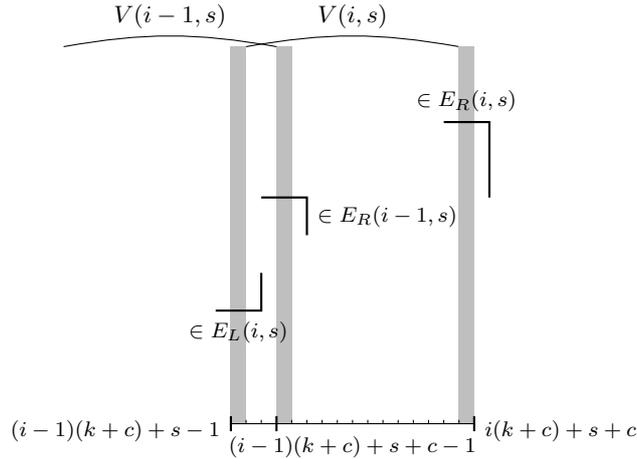

For $i \in \{0,r\}$, we now let $G_I(i,s) = G[I(i,s)]$ and let $G_{LR}(i,s)$ be the graph with vertex set $V(i, s)$ and edge set $E(G[V(i,s)]) \cup \{uv:u,v \in E(i,s)\}$. Moreover, for $0 < i < r$, we let
\begin{itemize}
\item $G_I(i,s) = G[I(i,s)]$;
\item $G_L(i,s)$ be the graph with vertex set $I_L(i,s) = I(i,s) \cup E_L(i,s)$ and edge set $E(G[I_L(i,s)]) \cup \{uv: u,v \in E_L(i,s)\}$;
\item $G_R(i,s)$ be the graph with vertex set $I_R(i,s) = I(i,s) \cup E_R(i,s)$ and edge set $E(G[I_R(i,s)]) \cup \{uv: u,v \in E_R(i,s)\}$;
\item $G_{LR}(i,s)$ be the graph with vertex set $V(i,s)$ and edge set $E(G[V(i,s)]) \cup \{uv: u,v \in E_L(i,s)\} \cup \{uv: u,v \in E_R(i,s)\}$. 
\end{itemize}

Observe that, for any $i \in \{0,\ldots,r\}$, the VPG representations of $G_I(i,s)$ and $G[V(i,s)]$ induced by $\mathcal{R}$ contain at most $k + 2c$ and $k + 4c$ columns, respectively. Therefore, by \Cref{clm:induced}, \Cref{lem:boundedmimvpg} and \Cref{lem:cliquemimwidth}, we can compute in polynomial time a branch decomposition with bounded mim-width for each of $G_I(i,s)$, $G[V(i,s)]$ and $G_{LR}(i,s)$. Similarly, for any $0 < i < r$, the VPG representations of $G[I_L(i,s)]$ and $G[I_R(i,s)]$ induced by $\mathcal{R}$ contain at most $k + 3c$ columns. Therefore, by \Cref{clm:induced}, \Cref{lem:boundedmimvpg} and \Cref{lem:cliquemimwidth}, we can compute in polynomial time a branch decomposition with bounded mim-width for each of $G_L(i,s)$ and $G_R(i,s)$. We then run the algorithm in \Cref{algomim} with these branch decompositions to compute (when the corresponding graph exists) minimum dominating sets $S_I(i,s)$, $S_L(i,s)$, $S_R(i,s)$ and $S_{LR}(i,s)$ of $G_I(i,s)$, $G_L(i,s)$, $G_R(i,s)$ and $G_{LR}(i,s)$, respectively, in polynomial time. Now, for $0 < i < r$, let $S(i,s)$ be a set of minimum cardinality among $S_I(i,s)$, $S_L(i,s)$, $S_R(i,s)$ and $S_{LR}(i,s)$. Similarly, for $i \in \{0, r\}$, let $S(i,s)$ be a set of minimum cardinality among $S_I(i,s)$ and $S_{LR}(i,s)$. We then repeat the procedure above for each fixed $s \in \{0, \ldots, k + c - 1\}$ in order to compute the sets $S_s = \bigcup_{i=0}^r S(i,s)$ and return the smallest set $S$ among the $S_s$'s in polynomial time. We will show that $S$ is a dominating set of $G$ such that $\vert S \vert \leq (1+\varepsilon) \vert \mathsf{OPT} \vert$, where $\mathsf{OPT}$ denotes an optimal solution for \textsc{Dominating Set} with instance $G$, thus concluding the proof. 

\begin{claim}\label{claimproof}
For any $s \in \{0, \ldots, k + c - 1\}$, $S_s = \bigcup_{i=0}^r S(i,s)$ is a dominating set of $G$. In particular, $S$ is a dominating set of $G$.
\end{claim}

Let $s \in \{0, \ldots, k + c - 1\}$ be fixed. We show that for any vertex $v \in V(G)$, there exists $i \in \{0,\ldots,r\}$ such that $v$ is an interior vertex of $V(i,s)$ i.e., $v \in I(i,s)$. Since $S(i,s)$ dominates every interior vertex of $V(i,s)$ by construction, \Cref{claimproof} would follow. Consider a vertex $v \in V(G)$ and let $p$ and $q$ denote the smallest and largest index $j \in \{0, \dots, m - 1\}$ such that $v \in V_j$, respectively. By definition, $v \notin X_{p-1} \cup X_{q}$. Since the horizontal part of $P_{v}$ has length at most $c$, we have $q - p \leq c$. 

Suppose first that $0 \leq p \leq s + c - 1$. If $q \leq s + c - 1$, then $v$ is an interior vertex of $V(0,s)$. Otherwise, $q > s + c - 1$, which implies that $p \geq s$ and so $v$ is an interior vertex of $V(1,s)$. Suppose now that $p > s + c - 1$ and consider the largest index $i \in \{1,\ldots,r\}$ such that $(i-1) \cdot (k + c) + s \leq p$. If $i = r$, then $v$ is an interior vertex of $V(r,s)$. Otherwise, $0 < i < r$ and by maximality of $i$ we have that $p < i \cdot (k + c) + s$, which implies that $q  <  i \cdot (k + c) + s + c$, and so $v$ is an interior vertex of $V(i,s)$. $\lozenge$

\begin{claim}\label{claim2}
For any $i \in \{0,\ldots,r\}$, $\vert S(i,s) \vert \leq \vert \mathsf{OPT} \cap V(i,s) \vert$.
\end{claim}

Since no vertex in $V(G) \setminus V(i,s)$ can dominate an interior vertex of $V(i,s)$ and $\mathsf{OPT}$ dominates every vertex in $G$, it follows that $\mathsf{OPT}(i,s)=\mathsf{OPT} \cap V(i,s)$ dominates every interior vertex of $V(i,s)$. In the following, we assume that $0 < i < r$ (the cases $i=0$ and $i=r$ can be treated similarly). We distinguish cases depending on whether $\mathsf{OPT}(i,s)$ contains vertices of $E_L(i,s)$ and $E_R(i,s)$.\\

\noindent
\textbf{Case 1.} \textit{$\mathsf{OPT}(i,s)$ contains no vertex of $E_L(i,s)$ and no vertex of $E_R(i,s)$.} As observed above, $\mathsf{OPT}(i,s)$ is a dominating set of $G_I(i,s)$ and so $\vert S(i,s) \vert \leq \gamma (G_{I}(i,s)) \leq \vert \mathsf{OPT}(i,s) \vert$.\\

\noindent
\textbf{Case 2.} \textit{$\mathsf{OPT}(i,s)$ contains a vertex of $E_L(i,s)$ and a vertex of $E_R(i,s)$.} Then, $\mathsf{OPT}(i,s)$ is a dominating set of $G_{LR}(i,s)$, as $E_L(i,s)$ and $E_R(i,s)$ are cliques in $G_{LR}(i,s)$. It follows that $\vert S(i,s) \vert \leq \gamma (G_{LR}(i,s)) \leq \vert \mathsf{OPT}(i,s) \vert$.\\

\noindent
\textbf{Case 3.} \textit{$\mathsf{OPT}(i,s)$ contains a vertex of either $E_L(i,s)$ or $E_R(i,s)$ but not of both.} Assume without loss of generality that $\mathsf{OPT}(i,s)$ contains a vertex of $E_L(i,s)$ (the other case is symmetric). Then, $\mathsf{OPT}(i,s)$ is a dominating set of $G_L(i,s)$, as $E_L(i,s)$ is a clique in $G_L(i,s)$. It follows that $\vert S(i,s) \vert \leq \gamma (G_{L}(i,s)) \leq \vert \mathsf{OPT}(i,s) \vert$.\\

\noindent In any case, we have that $\vert S(i,s) \vert \leq \vert \mathsf{OPT}(i,s) \vert$, as claimed. $\lozenge$\\

In order to conclude the proof, it is then enough to show that $\vert S \vert \leq (1+\varepsilon) \vert \mathsf{OPT} \vert$. By \Cref{claim2} we have that, for any $s \in \{0,\ldots,k + c -1\}$,

\[\vert S_s \vert \leq \sum_{i = 0}^{r} \vert S(i,s) \vert \leq \sum_{i=0}^{r} \vert \mathsf{OPT} \cap V(i,s) \vert. \]
It then follows from \Cref{obs:Vis} that 
  
\[ \sum_{i=0}^{r} \vert \mathsf{OPT} \cap V(i,s) \vert \leq \vert \mathsf{OPT} \vert + \sum_{i=0}^{r-1} \vert \mathsf{OPT} \cap V(i,s) \cap V(i+1,s) \vert. \]
On the other hand, since the horizontal part of each path has length at most $c$ and $V(i,s) \cap V(i+1,s) = \bigcup_{\ell = 0}^{c-1} V_{i \cdot (k + c) + s + \ell}$, we have that
\[\sum_{s = 0}^{k + c -1} \sum_{i=0}^{r-1} \vert \mathsf{OPT} \cap V(i,s) \cap V(i+1,s) \vert  \leq 2c\vert \mathsf{OPT} \vert, \]
which implies that
\[(k + c) \cdot \min_{s \in \{0, \ldots, k + c -1\}} \sum_{i=0}^{r-1} \vert \mathsf{OPT} \cap V(i,s) \cap V(i+1,s) \vert  \leq 2c \vert \mathsf{OPT} \vert.\]
Denoting by $s_0$ the index attaining the minimum above and combining the previous inequalities we have, 

\begin{align*}\vert S \vert \leq |S_{s_0}| \leq \sum_{i=0}^{r} \vert \mathsf{OPT} \cap V(i,s_0) \vert &\leq \vert \mathsf{OPT} \vert + \sum_{i=0}^{r-1} \vert \mathsf{OPT} \cap V(i,s_0) \cap V(i+1,s_0)\vert \\ &\leq \vert \mathsf{OPT} \vert + \frac{2c}{k + c} \vert \mathsf{OPT} \vert \\
&\leq (1+\varepsilon) \vert \mathsf{OPT} \vert,
\end{align*}
 thus concluding the proof. 
\end{proof}

\begin{remark}\label{best} \Cref{thm:ptasDS} is best possible in the sense that, if we remove one of conditions \ref{2dom} and \ref{3dom}, \textsc{Dominating Set} does not admit a PTAS, unless $\mathsf{P} = \mathsf{NP}$. 

Indeed, every split graph admits a VPG representation $\mathcal{R} = (\mathcal{G}, \mathcal{P})$ such that each path in $\mathcal{P}$ has $O(\vert \mathcal{P} \vert)$ bends and horizontal part of length at most $3$ (see \Cref{splitrep}). On the other hand, \textsc{Dominating Set} restricted to split graphs cannot be approximated to within a factor of $(1 - \varepsilon)\ln n$, for any constant $\varepsilon > 0$, unless $\mathsf{NP} \subseteq \mathsf{DTIME}(n^{O(\log\log n)})$ \citep{CC08}.

Moreover, every circle graph is a $1$-string $B_{1}$-VPG graph \citep{asinowski} and \textsc{Dominating Set} is $\mathsf{APX}$-hard on circle graphs \citep{DP06}. 
\end{remark}

\section{Concluding remarks and open problems}

In this paper we showed that \textsc{Independent Set} and \textsc{Dominating Set} admit PTASes on VPG graphs with a representation in which each path has polynomially many bends (in particular, on $B_{k}$-VPG graphs, for fixed $k \geq 0$) if in addition each grid-edge belongs to a bounded number of paths and the horizontal part of each path is bounded. Moreover, in the case of \textsc{Dominating Set}, we observed that this is not true if we remove any of the two constraints, unless $\mathsf{P} = \mathsf{NP}$. On the other hand, the situation remains obscure in the case of \textsc{Independent Set}: Does \Cref{thm:ptasIS} still hold if we remove one of conditions \ref{2nd} and \ref{3rd}? On a similar note, as already mentioned in \Cref{intro}, it is a major open problem to determine whether \textsc{Independent Set} admits constant-factor approximation algorithms on $B_{k}$-VPG graphs.      

In \Cref{sechard}, we showed that \textsc{Independent Set} and \textsc{Dominating Set} are $\mathsf{NP}$-complete when restricted to $B_0$-CPG graphs admitting a representation where each horizontal path has length at most $2$. It is then natural to ask what happens if horizontal and vertical paths all have length $1$. It is easy to see that a unit $B_{0}$-VPG graph is claw-free and so \textsc{Independent Set} is polynomial-time solvable \citep{Sbi80}. On the other hand, the following remains open: What is the complexity of \textsc{Dominating Set} for unit $B_0$-CPG graphs? In case the problem is in $\mathsf{P}$, is the same true for the superclass of unit $B_{0}$-VPG graphs?

%=================================================================

%=============================

\bibliographystyle{plainnat}
\bibliography{references}

\end{document}